\documentclass[12pt,reqno]{amsart}
\usepackage{graphicx}
\usepackage{amscd}
\usepackage{amsmath}
\usepackage{epsfig}
\usepackage{amsfonts}
\usepackage{amssymb}

\providecommand{\U}[1]{\protect\rule{.1in}{.1in}}
\providecommand{\U}[1]{\protect\rule{.1in}{.1in}}
\allowdisplaybreaks

\theoremstyle{plain}

\newtheorem{lemma}{Lemma}

\newtheorem{remark}{Remark}

\numberwithin{equation}{section}

\input{tcilatex}

\begin{document}
\title[The Riccati System]{The Riccati System and a Diffusion-Type Equation}
\author{Erwin Suazo}
\address{Department of Mathematical Sciences, University of Puerto Rico,
Mayaguez, Call Box 9000, Puerto Rico 00681--9000.}
\email{erwin.suazo@upr.edu}
\author{Sergei K. Suslov}
\address{School of Mathematical and Statistical Sciences \& Mathematical,
Computational and Modeling Sciences Center, Arizona State University, Tempe,
AZ 85287--1804, U.S.A.}
\email{sks@asu.edu}
\urladdr{http://hahn.la.asu.edu/\symbol{126}suslov/index.html}
\author{Jos\'{e} M. Vega-Guzm\'{a}n}
\address{Mathematical, Computational and Modeling Sciences Center, Arizona
State University, Tempe, AZ 85287--1904, U.S.A.}
\email{jmvega@asu.edu}
\date{\today }

\begin{abstract}
We discuss a method of constructing solution of the initial value problem
for diffusion-type equations in terms of solutions of certain Riccati and
Ermakov-type systems. A nonautonomous Burgers-type equation is also
considered.
\end{abstract}

\maketitle

\section{Introduction}

A goal of this note, complementary to our recent paper \cite{SuazoSusVega10}%
, is to elaborate on the Cauchy initial value problem for a class of
nonautonomous and inhomogeneous diffusion-type equations on $%
%TCIMACRO{\U{211d} }%
%BeginExpansion
\mathbb{R}
%EndExpansion
.$ A corresponding nonautonomous Burgers-type equation is also analyzed as a
by-product. Here, we use explicit transformations to the standard forms and
emphasize natural relations with certain Riccati and Ermakov-type systems,
which seem are missing in the available literature. Similar methods are
applied to the corresponding Schr\"{o}dinger equation (see, for example, 
\cite{Cor-Sot:Lop:Sua:Sus}, \cite{Cor-Sot:Sua:Sus}, \cite{Cor-Sot:Sua:SusInv}%
, \cite{Cor-Sot:Sus}, \cite{Cor-Sot:SusDPO}, \cite{Lan:Lop:Sus}, \cite%
{Lan:Sus}, \cite{Lop:Sus}, \cite{Me:Co:Su}, \cite{Suaz:Sus}, \cite{Suslov10}%
, \cite{Suslov11} and references therein). A group theoretical approach to a
similar class of partial differential equations is discussed in Refs.~\cite%
{GagWint93}, \cite{Miller77} and \cite{Rosen76}.

For an introduction to fundamental solutions for parabolic equations, see
chapter one of the book by Friedman \cite{Friedman64}. Among numerous
applications, we only elaborate here on an important role of fundamental
solutions in probability theory \cite{Craddock09}, \cite{KaratzShreve}.
Consider an It\^{o} diffusion $X=\left\{ X_{t}:t\geq 0\right\} $ which
satisfies the stochastic differential equation%
\begin{equation}
dX_{t}=b\left( X_{t},t\right) \ dt+\sigma \left( X_{t},t\right) \
dW_{t},\qquad X_{0}=x,  \label{SDE}
\end{equation}%
in which $W=\left\{ W_{t}:t\geq 0\right\} $ is a standard Wiener process.
The existence and uniqueness of solutions of (\ref{SDE}) depends on the
coefficients $b$ and $\sigma .$ (See Ref.~\cite{KaratzShreve} for conditions
of unique strong solution to (\ref{SDE}).) If the equation (\ref{SDE}) has a
unique solution, then the expectations%
\begin{equation}
u\left( x,t\right) =E_{x}\left[ \phi \left( X_{t}\right) \right] =E\left[
\phi \left( X_{t}\right) |X_{0}=x\right]  \label{Expect}
\end{equation}%
are solutions of the Cauchy problem%
\begin{equation}
u_{t}=\frac{1}{2}\sigma ^{2}\left( x,t\right) u_{xx}+b\left( x,t\right)
u_{x},\qquad u\left( x,0\right) =\phi \left( x\right) .  \label{SDECauchy}
\end{equation}%
This PDE is known as Kolmogorov forward equation \cite{Craddock09}, \cite%
{KaratzShreve}. Thus if $p\left( x,y,t\right) $ is the appropriate
fundamental solution of (\ref{SDECauchy}), then one can compute the given
expectations according to%
\begin{equation}
E_{x}\left[ \phi \left( X_{t}\right) \right] =\int_{\Omega }p\left(
x,y,t\right) \phi \left( y\right) \ dy.  \label{SDEExpect}
\end{equation}%
In this context, the fundamental solution is known as the probability
transition density for the process and%
\begin{equation}
\int_{\Omega }p\left( x,y,t\right) \ dy=1.  \label{norm}
\end{equation}%
See also Refs.~\cite{Albev:Roz10} and \cite{Kamb10} for applications to
stochastic differential equations related to Fokker--Planck and Burgers
equations.

\section{Transformation to the Standard Form}

We present the following result.

\begin{lemma}
The nonautonomous and inhomogeneous diffusion-type equation%
\begin{equation}
\frac{\partial u}{\partial t}=a\left( t\right) \frac{\partial ^{2}u}{%
\partial x^{2}}-\left( g\left( t\right) -c\left( t\right) x\right) \frac{%
\partial u}{\partial x}+\left( d\left( t\right) +f\left( t\right) x-b\left(
t\right) x^{2}\right) u,  \label{heat}
\end{equation}%
where $a,b,c,d,f,g$ are suitable functions of time $t$\ only, can be reduced
to the standard autonomous form%
\begin{equation}
\frac{\partial v}{\partial \tau }=\frac{\partial ^{2}v}{\partial \xi ^{2}}
\label{standard}
\end{equation}%
with the help of the following substitution:%
\begin{align}
u\left( x,t\right) & =\frac{1}{\sqrt{\mu \left( t\right) }}e^{\alpha \left(
t\right) x^{2}+\delta \left( t\right) x+\kappa \left( t\right) }v\left( \xi
,\tau \right) ,  \label{substitution} \\
\xi & =\beta \left( t\right) x+\varepsilon \left( t\right) ,\quad \tau
=\gamma \left( t\right) .  \notag
\end{align}%
Here, $\mu ,\alpha ,\beta ,\gamma ,\delta ,\varepsilon ,\kappa $ are
functions of $t$\ that satisfy%
\begin{equation}
\frac{\mu ^{\prime }}{2\mu }+2a\alpha +d=0  \label{mu}
\end{equation}%
and%
\begin{align}
& \frac{d\alpha }{dt}+b-2c\alpha -4a\alpha ^{2}=0,  \label{alpha} \\
& \frac{d\beta }{dt}-\left( c+4a\alpha \right) \beta =0,  \label{beta} \\
& \frac{d\gamma }{dt}-a\beta ^{2}=0,  \label{gamma} \\
& \frac{d\delta }{dt}-\left( c+4a\alpha \right) \delta =f-2\alpha g,
\label{delta} \\
& \frac{d\varepsilon }{dt}+\left( g-2a\delta \right) \beta =0,
\label{epsilon} \\
& \frac{d\kappa }{dt}+g\delta -a\delta ^{2}=0.  \label{kappa}
\end{align}
\end{lemma}

Equation (\ref{alpha}) is called the \textit{Riccati nonlinear differential
equation} \cite{Reid72}, \cite{Wa}, \cite{Wh:Wa} and we shall refer to the
system (\ref{alpha})--(\ref{kappa}) as a \textit{Riccati-type system}.

The substitution (\ref{mu}) reduces the nonlinear Riccati equation (\ref%
{alpha}) to the second order linear equation%
\begin{equation}
\mu ^{\prime \prime }-\tau \left( t\right) \mu ^{\prime }-4\sigma \left(
t\right) \mu =0,  \label{charequation}
\end{equation}%
where%
\begin{equation}
\tau \left( t\right) =\frac{a^{\prime }}{a}+2c-4d,\quad \sigma \left(
t\right) =ab+cd-d^{2}+\frac{d}{2}\left( \frac{a^{\prime }}{a}-\frac{%
d^{\prime }}{d}\right) ,  \label{SigTau}
\end{equation}%
which shall be referred to as a \textit{characteristic equation} \cite%
{SuazoSusVega10}.

It is also known \cite{SuazoSusVega10} that the diffusion-type equation (\ref%
{heat}) has a particular solution of the form%
\begin{equation}
u=\frac{1}{\sqrt{\mu \left( t\right) }}e^{\alpha \left( t\right) x^{2}+\beta
\left( t\right) xy+\gamma \left( t\right) y^{2}+\delta \left( t\right)
x+\varepsilon \left( t\right) y+\kappa \left( t\right) },
\label{partsolution}
\end{equation}%
provided that the time dependent functions $\mu ,\alpha ,\beta ,\gamma
,\delta ,\varepsilon ,\kappa $ satisfy the Riccati-type system (\ref{mu})--(%
\ref{kappa}).

A group theoretical approach to a similar class of partial differential
equations is discussed in Refs.~\cite{GagWint93}, \cite{Miller77} and \cite%
{Rosen76}.

\section{Fundamental Solution}

By the \textit{superposition principle} one can solve (formally) the Cauchy
initial value problem for the diffusion-type equation (\ref{heat}) subject
to initial data $u\left( x,0\right) =\varphi \left( x\right) $ on the entire
real line $-\infty <x<\infty $ in an integral form%
\begin{equation}
u\left( x,t\right) =\int_{-\infty }^{\infty }K_{0}\left( x,y,t\right) \
\varphi \left( x\right) dy  \label{CIVP}
\end{equation}%
with the \textit{fundamental solution} (heat kernel) \cite{SuazoSusVega10}:%
\begin{equation}
K_{0}\left( x,y,t\right) =\frac{1}{\sqrt{2\pi \mu _{0}\left( t\right) }}\
e^{\alpha _{0}\left( t\right) x^{2}+\beta _{0}\left( t\right) xy+\gamma
_{0}\left( t\right) y^{2}+\delta _{0}\left( t\right) x+\varepsilon
_{0}\left( t\right) y+\kappa _{0}\left( t\right) },  \label{heatkernel}
\end{equation}%
where a particular solution of the Riccati-type system (\ref{alpha})--(\ref%
{kappa}) is given by 
\begin{equation}
\alpha _{0}\left( t\right) =-\frac{1}{4a\left( t\right) }\frac{\mu
_{0}^{\prime }\left( t\right) }{\mu _{0}\left( t\right) }-\frac{d\left(
t\right) }{2a\left( t\right) },  \label{A0}
\end{equation}%
\begin{equation}
\beta _{0}\left( t\right) =\frac{h\left( t\right) }{\mu _{0}\left( t\right) }%
,\quad h\left( t\right) =\exp \left( \int_{0}^{t}\left( c\left( s\right)
-2d\left( s\right) \right) \ ds\right) ,  \label{B0}
\end{equation}%
\begin{align}
\gamma _{0}\left( t\right) & =\frac{d\left( 0\right) }{2a\left( 0\right) }-%
\frac{a\left( t\right) h^{2}\left( t\right) }{\mu _{0}\left( t\right) \mu
_{0}^{\prime }\left( t\right) }-4\int_{0}^{t}\frac{a\left( s\right) \sigma
\left( s\right) h\left( s\right) }{\left( \mu _{0}^{\prime }\left( s\right)
\right) ^{2}}\ ds  \label{C0} \\
& =\frac{d\left( 0\right) }{2a\left( 0\right) }-\frac{1}{2\mu _{1}\left(
0\right) }\frac{\mu _{1}\left( t\right) }{\mu _{0}\left( t\right) },
\label{C1}
\end{align}%
\begin{equation}
\delta _{0}\left( t\right) =\frac{h\left( t\right) }{\mu _{0}\left( t\right) 
}\ \ \int_{0}^{t}\left[ \left( f\left( s\right) +\frac{d\left( s\right) }{%
a\left( s\right) }g\left( s\right) \right) \mu _{0}\left( s\right) +\frac{%
g\left( s\right) }{2a\left( s\right) }\mu _{0}^{\prime }\left( s\right) %
\right] \ \frac{ds}{h\left( s\right) },  \label{D0}
\end{equation}%
\begin{align}
\varepsilon _{0}\left( t\right) & =-\frac{2a\left( t\right) h\left( t\right) 
}{\mu _{0}^{\prime }\left( t\right) }\delta _{0}\left( t\right)
-8\int_{0}^{t}\frac{a\left( s\right) \sigma \left( s\right) h\left( s\right) 
}{\left( \mu _{0}^{\prime }\left( s\right) \right) ^{2}}\left( \mu
_{0}\left( s\right) \delta _{0}\left( s\right) \right) \ ds  \label{E0} \\
& \quad +2\int_{0}^{t}\frac{a\left( s\right) h\left( s\right) }{\mu
_{0}^{\prime }\left( s\right) }\left[ f\left( s\right) +\frac{d\left(
s\right) }{a\left( s\right) }g\left( s\right) \right] \ ds,  \notag
\end{align}%
\begin{align}
\kappa _{0}\left( t\right) & =-\frac{a\left( t\right) \mu _{0}\left(
t\right) }{\mu _{0}^{\prime }\left( t\right) }\delta _{0}^{2}\left( t\right)
-4\int_{0}^{t}\frac{a\left( s\right) \sigma \left( s\right) }{\left( \mu
_{0}^{\prime }\left( s\right) \right) ^{2}}\left( \mu _{0}\left( s\right)
\delta _{0}\left( s\right) \right) ^{2}\ ds  \label{K0} \\
& \quad +2\int_{0}^{t}\frac{a\left( s\right) }{\mu _{0}^{\prime }\left(
s\right) }\left( \mu _{0}\left( s\right) \delta _{0}\left( s\right) \right) %
\left[ f\left( s\right) +\frac{d\left( s\right) }{a\left( s\right) }g\left(
s\right) \right] \ ds  \notag
\end{align}%
with $\delta \left( 0\right) =g\left( 0\right) /\left( 2a\left( 0\right)
\right) ,$ $\varepsilon \left( 0\right) =-\delta \left( 0\right) ,$ $\kappa
\left( 0\right) =0.$ Here, $\mu _{0}$ and $\mu _{1}$ are the so-called 
\textit{standard solutions} of the characteristic equation (\ref%
{charequation}) subject to the\ following initial data%
\begin{equation}
\mu _{0}\left( 0\right) =0,\quad \mu _{0}^{\prime }\left( 0\right) =2a\left(
0\right) \neq 0\qquad \mu _{1}\left( 0\right) \neq 0,\quad \mu _{1}^{\prime
}\left( 0\right) =0.  \label{standarddata}
\end{equation}%
Solution (\ref{A0})--(\ref{K0}) shall be referred to as a \textit{%
fundamental solution} of the Riccati-type system (\ref{alpha})--(\ref{kappa}%
); see (\ref{AsA0})--(\ref{AsF0}) and (\ref{AsK0}) for the corresponding
asymptotics.

\begin{lemma}
The Riccati-type system (\ref{mu})--(\ref{kappa}) has the following
(general) solution:%
\begin{align}
& \mu \left( t\right) =-2\mu \left( 0\right) \mu _{0}\left( t\right) \left(
\alpha \left( 0\right) +\gamma _{0}\left( t\right) \right) ,  \label{MKernel}
\\
& \alpha \left( t\right) =\alpha _{0}\left( t\right) -\frac{\beta
_{0}^{2}\left( t\right) }{4\left( \alpha \left( 0\right) +\gamma _{0}\left(
t\right) \right) },  \label{AKernel} \\
& \beta \left( t\right) =-\frac{\beta \left( 0\right) \beta _{0}\left(
t\right) }{2\left( \alpha \left( 0\right) +\gamma _{0}\left( t\right)
\right) },  \label{BKernel} \\
& \gamma \left( t\right) =\gamma \left( 0\right) -\frac{\beta ^{2}\left(
0\right) }{4\left( \alpha \left( 0\right) +\gamma _{0}\left( t\right)
\right) }  \label{CKernel}
\end{align}%
and%
\begin{align}
\delta \left( t\right) & =\delta _{0}\left( t\right) -\frac{\beta _{0}\left(
t\right) \left( \delta \left( 0\right) +\varepsilon _{0}\left( t\right)
\right) }{2\left( \alpha \left( 0\right) +\gamma _{0}\left( t\right) \right) 
},  \label{DKernel} \\
\varepsilon \left( t\right) & =\varepsilon \left( 0\right) -\frac{\beta
\left( 0\right) \left( \delta \left( 0\right) +\varepsilon _{0}\left(
t\right) \right) }{2\left( \alpha \left( 0\right) +\gamma _{0}\left(
t\right) \right) },  \label{EKernel} \\
\kappa \left( t\right) & =\kappa \left( 0\right) +\kappa _{0}\left( t\right)
-\frac{\left( \delta \left( 0\right) +\varepsilon _{0}\left( t\right)
\right) ^{2}}{4\left( \alpha \left( 0\right) +\gamma _{0}\left( t\right)
\right) }  \label{FKernel}
\end{align}%
in terms of the fundamental solution (\ref{A0})--(\ref{K0}) subject to
arbitrary initial data $\mu \left( 0\right) ,$ $\alpha \left( 0\right) ,$ $%
\beta \left( 0\right) ,$ $\gamma \left( 0\right) ,$ $\delta \left( 0\right)
, $ $\varepsilon \left( 0\right) ,$ $\kappa \left( 0\right) .$
\end{lemma}

\begin{proof}
Use (\ref{partsolution})--(\ref{heatkernel}), uniqueness of the solution and
the elementary integral:%
\begin{equation}
\int_{-\infty }^{\infty }e^{-ay^{2}+2by}\ dy=\sqrt{\frac{\pi }{a}}\
e^{b^{2}/a},\quad a>0.  \label{Gauss}
\end{equation}%
Computational details are left to the reader.
\end{proof}

\begin{remark}
It is worth noting that our transformation (\ref{substitution}), combined
with the standard heat kernel \cite{NiPDE}:%
\begin{equation}
K_{0}\left( \xi ,\eta ,\tau \right) =\frac{1}{\sqrt{4\pi \left( \tau -\tau
_{0}\right) }}\exp \left[ -\frac{\left( \xi -\eta \right) ^{2}}{4\left( \tau
-\tau _{0}\right) }\right]  \label{heatstandard}
\end{equation}%
for the diffusion equation (\ref{standard}) and (\ref{MKernel})--(\ref%
{FKernel}), allows one to derive the fundamental solution (\ref{heatkernel})
of the diffusion-type equation (\ref{heat}) from a new perspective.
\end{remark}

\begin{lemma}
Solution (\ref{MKernel})--(\ref{FKernel}) implies:%
\begin{align}
& \mu _{0}=\frac{2\mu }{\mu \left( 0\right) \beta ^{2}\left( 0\right) }%
\left( \gamma -\gamma \left( 0\right) \right) ,  \label{M0} \\
& \alpha _{0}=\alpha _{0}\left( t\right) -\frac{\beta ^{2}}{4\left( \gamma
-\gamma \left( 0\right) \right) },  \label{AA0} \\
& \beta _{0}=\frac{\beta \left( 0\right) \beta }{2\left( \gamma -\gamma
\left( 0\right) \right) },  \label{BB0} \\
& \gamma _{0}=-\alpha \left( 0\right) -\frac{\beta ^{2}\left( 0\right) }{%
4\left( \gamma -\gamma \left( 0\right) \right) }  \label{CC0}
\end{align}%
and%
\begin{align}
\delta _{0}& =\delta -\frac{\beta \left( \varepsilon -\varepsilon \left(
0\right) \right) }{2\left( \gamma -\gamma \left( 0\right) \right) },
\label{DD0} \\
\varepsilon _{0}& =-\delta \left( 0\right) +\frac{\beta \left( 0\right)
\left( \varepsilon -\varepsilon \left( 0\right) \right) }{2\left( \gamma
-\gamma \left( 0\right) \right) },  \label{EE0} \\
\kappa _{0}& =\kappa -\kappa \left( 0\right) -\frac{\left( \varepsilon
-\varepsilon \left( 0\right) \right) ^{2}}{4\left( \gamma -\gamma \left(
0\right) \right) },  \label{FF0}
\end{align}%
which gives the following asymptotics%
\begin{align}
& \alpha _{0}\left( t\right) =-\frac{1}{4a\left( 0\right) t}-\frac{c\left(
0\right) }{4a\left( 0\right) }+\frac{a^{\prime }\left( 0\right) }{%
8a^{2}\left( 0\right) }+\mathcal{O}\left( t\right) ,  \label{AsA0} \\
& \beta _{0}\left( t\right) =\frac{1}{2a\left( 0\right) t}-\frac{a^{\prime
}\left( 0\right) }{4a^{2}\left( 0\right) }+\mathcal{O}\left( t\right) ,
\label{AsB0} \\
& \gamma _{0}\left( t\right) =-\frac{1}{4a\left( 0\right) t}+\frac{c\left(
0\right) }{4a\left( 0\right) }+\frac{a^{\prime }\left( 0\right) }{%
8a^{2}\left( 0\right) }+\mathcal{O}\left( t\right) ,  \label{AsC0} \\
& \delta _{0}\left( t\right) =\frac{g\left( 0\right) }{2a\left( 0\right) }+%
\mathcal{O}\left( t\right) ,\qquad \varepsilon _{0}\left( t\right) =-\frac{%
g\left( 0\right) }{2a\left( 0\right) }+\mathcal{O}\left( t\right) ,
\label{AsDE0} \\
& \kappa _{0}\left( t\right) =\mathcal{O}\left( t\right)  \label{AsF0}
\end{align}%
as $t\rightarrow 0^{+}.$
\end{lemma}

(The proof is left to the reader.)

These formulas allows to establish a required asymptotic of the fundamental
solution (\ref{heatkernel}):%
\begin{align}
K_{0}\left( x,y,t\right) & \sim \frac{1}{\sqrt{4\pi a\left( 0\right) t}}\exp %
\left[ -\frac{\left( x-y\right) ^{2}}{4a\left( 0\right) t}\right]
\label{AsK0} \\
& \times \exp \left[ \frac{a^{\prime }\left( 0\right) }{8a^{2}\left(
0\right) }\left( x-y\right) ^{2}-\frac{c\left( 0\right) }{4a\left( 0\right) }%
\left( x^{2}-y^{2}\right) \right] \exp \left[ \frac{g\left( 0\right) }{%
2a\left( 0\right) }\left( x-y\right) \right] .  \notag
\end{align}%
(Here, $f\sim g$ as $t\rightarrow 0^{+},$ if $\lim_{t\rightarrow
0^{+}}\left( f/g\right) =$ $1.$ The proof is left to the reader.)

By a direct substitution one can verify that the right hand sides of (\ref%
{MKernel})--(\ref{FKernel}) satisfy the Riccati-type system (\ref{mu})--(\ref%
{kappa}) and that the asymptotics (\ref{AsA0})--(\ref{AsF0}) result in the
continuity with respect to initial data:%
\begin{equation}
\lim_{t\rightarrow 0^{+}}\mu \left( t\right) =\mu \left( 0\right) ,\quad
\lim_{t\rightarrow 0^{+}}\alpha \left( t\right) =\alpha \left( 0\right)
,\quad \text{etc.\label{lims}}
\end{equation}%
The transformation property (\ref{MKernel})--(\ref{FKernel}) allows one to
find solution of the initial value problem in terms of the fundamental
solution (\ref{A0})--(\ref{K0}) and may be referred to as a \textit{%
nonlinear superposition principle} for the Riccati-type system.

\section{Eigenfunction Expansion and Ermakov-type System}

With the help of transformation (\ref{substitution}) one can reduce the
diffusion equation (\ref{heat}) to another convenient form%
\begin{equation}
\frac{\partial v}{\partial \tau }=\frac{\partial ^{2}v}{\partial \xi ^{2}}%
+\xi ^{2}v,  \label{harmonic}
\end{equation}%
which allows to find solution of the Cauchy initial value problem in terms
of an eigenfunction expansion similar to the case of the corresponding Schr%
\"{o}dinger in Refs.~\cite{Lan:Lop:Sus} and \cite{Suslov10}. This method
requires an extension the Riccati-type system (\ref{alpha})--(\ref{kappa})
to a more general Ermakov-type system \cite{Lan:Lop:Sus}, which is
integrable in quadratures once again in terms of solutions of the
characteristic equation (\ref{charequation}). Further details are left to
the reader.

\section{Nonautonomous Burgers Equation}

The nonlinear equation%
\begin{align}
& \frac{\partial v}{\partial t}+a\left( t\right) \left( v\frac{\partial v}{%
\partial x}-\frac{\partial ^{2}v}{\partial x^{2}}\right) -c\left( t\right)
\left( x\frac{\partial v}{\partial x}+v\right) +g\left( t\right) \frac{%
\partial v}{\partial x}  \label{NABurgers} \\
& \qquad =2\left( 2b\left( t\right) x-f\left( t\right) \right) ,  \notag
\end{align}%
when $a=1$ and $b=c=f=g=0,$ is known as \textit{Burgers' equation} \cite%
{Bateman15}, \cite{Burgers48}, \cite{Cole50}, \cite{Hopf50}, \cite%
{Kadom:Karp71}, \cite{Sach87}, \cite{Whitham}\textit{\ }and we shall refer
to (\ref{NABurgers}) as a nonautonomous \textit{Burgers-type equation}.

\begin{lemma}
The following identity holds%
\begin{align}
& v_{t}+a\left( vv_{x}-v_{xx}\right) +\left( g-cx\right) v_{x}-cv+2\left(
f-2bx\right)  \notag \\
& \qquad =-2\left( \frac{u_{t}-Qu}{u}\right) _{x},  \label{BurgersIdentity}
\end{align}%
if%
\begin{equation}
v=-2\frac{u_{x}}{u}\qquad \left( \text{The \textit{Cole--Hopf transformation}%
}\right)  \label{ColeHopf}
\end{equation}%
and%
\begin{equation}
Qu=au_{xx}-\left( g-cx\right) u_{x}+\left( d+fx-bx^{2}\right) u
\label{HOperator}
\end{equation}%
($a,$ $b,$ $c,$ $d,$ $f,$ $g$ are functions of $t$ only).
\end{lemma}

(This can be verified by a direct substitution.)

The substitution (\ref{ColeHopf}) turns the nonlinear Burgers-type equation (%
\ref{NABurgers}) into the diffusion-type equation (\ref{heat}). Then
solution of the corresponding Cauchy initial value problem can be
represented as%
\begin{equation}
v\left( x,t\right) =-2\frac{\partial }{\partial x}\ln \left[ \int_{-\infty
}^{\infty }K_{0}\left( x,y,t\right) \exp \left( -\frac{1}{2}%
\int_{0}^{y}v\left( z,0\right) \ dz\right) \ dy\right] ,  \label{CIVPBurgers}
\end{equation}%
where the heat kernel is given by (\ref{heatkernel}), for suitable initial
data $v\left( z,0\right) $ on $%
%TCIMACRO{\U{211d} }%
%BeginExpansion
\mathbb{R}
%EndExpansion
.$

\section{Traveling Wave Solutions of Burgers-type Equation}

Looking for solutions of our equation (\ref{NABurgers}) in the form%
\begin{equation}
v=\beta \left( t\right) F\left( \beta \left( t\right) x+\gamma \left(
t\right) \right) =\beta F\left( z\right) ,\quad z=\beta x+\gamma 
\label{BurgersSubst}
\end{equation}%
($\beta $ and $\gamma $ are functions of $t$ only), one gets%
\begin{equation}
F^{\prime \prime }=\left( c_{0}+c_{1}\right) F^{\prime }+FF^{\prime
}+2c_{2}z+c_{3}  \label{BurgersPreRiccati}
\end{equation}%
provided that%
\begin{eqnarray}
&&\beta ^{\prime }=c\beta ,\qquad \quad \gamma ^{\prime }=c_{0}a\beta ^{2},
\label{BSysA} \\
&&g=c_{1}a\beta ,\qquad b=-\frac{1}{2}c_{2}a\beta ^{4},  \label{BSysB} \\
&&f=\frac{1}{2}a\beta ^{3}\left( 2c_{2}\gamma +c_{3}\right)   \label{BSysC}
\end{eqnarray}%
($c_{0},$ $c_{1},$ $c_{2},$ $c_{3}$ are constants). From (\ref%
{BurgersPreRiccati}):%
\begin{equation}
F^{\prime }=\left( c_{0}+c_{1}\right) F+\frac{1}{2}%
F^{2}+c_{2}z^{2}+c_{3}z+c_{4},  \label{BurgersRiccati}
\end{equation}%
where $c_{4}$ is a constant of integration. The substitution%
\begin{equation}
F=-2\frac{\mu ^{\prime }}{\mu }  \label{RiccatiSubstitution}
\end{equation}%
transforms the Riccati equation (\ref{BurgersRiccati}) into a special case
of generalized equation of hypergeometric type:%
\begin{equation}
\mu ^{\prime \prime }-\left( c_{0}+c_{1}\right) \mu ^{\prime }+\frac{1}{2}%
\left( c_{2}z^{2}+c_{3}z+c_{4}\right) \mu =0,  \label{NUEquation}
\end{equation}%
which can be solved in general by methods of Ref.~\cite{Ni:Uv}. Elementary
solutions are discussed, for example, in \cite{KudryashovBook10} and \cite%
{Kudryash:Sine09}.

\section{Some Examples}

Now we consider from a united viewpoint several elementary diffusion and
Burgers equations that are important in applications.

\bigskip

\textbf{Example~1\quad }For the standard \textit{diffusion equation} on $%
%TCIMACRO{\U{211d} }%
%BeginExpansion
\mathbb{R}
%EndExpansion
:$%
\begin{equation}
\frac{\partial u}{\partial t}=a\frac{\partial ^{2}u}{\partial x^{2}},\qquad
a=\text{constant}>0  \label{sp1}
\end{equation}%
the heat kernel is given by%
\begin{equation}
K\left( x,y,t\right) =\frac{1}{\sqrt{4\pi at}}\exp \left[ -\frac{\left(
x-y\right) ^{2}}{4at}\right] ,\qquad t>0.  \label{sp2}
\end{equation}%
(See \cite{Cann}, \cite{NiPDE} and references therein for a detailed
investigation of the classical one-dimensional heat equation.)

\bigskip

\textbf{Example 2\quad }In mathematical description of the nerve cell a
dendritic branch is typically modeled by using cylindrical \textit{cable
equation} \cite{JackNobleTsien83}:%
\begin{equation}
\tau \frac{\partial u}{\partial t}=\lambda ^{2}\frac{\partial ^{2}u}{%
\partial x^{2}}+u,\quad \tau =\text{constant}>0.  \label{sp3}
\end{equation}%
The fundamental solution on $%
%TCIMACRO{\U{211d} }%
%BeginExpansion
\mathbb{R}
%EndExpansion
$ is given by%
\begin{equation}
K_{0}\left( x,y,t\right) =\frac{\sqrt{\tau }e^{t/\tau }}{\sqrt{4\pi \lambda
^{2}t}}\exp \left[ -\frac{\tau \left( x-y\right) ^{2}}{4\lambda ^{2}t}\right]
,\quad t>0.  \label{sp4}
\end{equation}%
(See also \cite{Herr-Val:Sus} and references therein.)

\bigskip

\textbf{Example 3\quad }The fundamental solution of the \textit{%
Fokker-Planck equation }\cite{Risken89}, \cite{Yau04}: 
\begin{equation}
\frac{\partial u}{\partial t}=\frac{\partial ^{2}u}{\partial x^{2}}+x\frac{%
\partial u}{\partial x}+u  \label{sp5}
\end{equation}%
on $%
%TCIMACRO{\U{211d} }%
%BeginExpansion
\mathbb{R}
%EndExpansion
$ is given by \cite{SuazoSusVega10}:%
\begin{equation}
K_{0}\left( x,y,t\right) =\frac{1}{\sqrt{2\pi \left( 1-e^{-2t}\right) }}\exp %
\left[ -\frac{\left( x-e^{-t}y\right) ^{2}}{2\left( 1-e^{-2t}\right) }\right]
,\quad t>0.  \label{sp6}
\end{equation}%
Here, 
\begin{equation}
\lim_{t\rightarrow \infty }K_{0}\left( x,y,t\right) =\frac{e^{-x^{2}/2}}{%
\sqrt{2\pi }},\qquad y=\text{constant}.  \label{FPLimit}
\end{equation}

\bigskip

\textbf{Example 4\quad }Equation%
\begin{equation}
\frac{\partial u}{\partial t}=a\frac{\partial ^{2}u}{\partial x^{2}}+\left(
g-kx\right) \frac{\partial u}{\partial x},\qquad a,k>0,\quad g\geq 0
\label{sp7}
\end{equation}%
corresponds to the heat equation with linear drift when $g=0$ \cite{Miller77}%
. In stochastic differential equations this equation corresponds the
Kolmogorov forward equation for the regular Ornstein--Uhlenbech process \cite%
{Craddock09}. The fundamental solution is given by%
\begin{align}
& K_{0}\left( x,y,t\right) =\frac{\sqrt{k}e^{kt/2}}{\sqrt{4\pi a\sinh \left(
kt\right) }}  \label{sp8} \\
& \qquad \times \exp \left[ -\frac{\left( k\left(
xe^{-kt/2}-ye^{kt/2}\right) +2g\sinh \left( kt/2\right) \right) ^{2}}{%
4ak\sinh \left( kt\right) }\right] ,\quad t>0.  \notag
\end{align}%
(See Refs.~\cite{Craddock09} and \cite{SuazoSusVega10} for more details.)

\bigskip

\textbf{Example 5\quad }The \textit{viscous Burgers equation }\cite%
{Bateman15}, \cite{Burgers48}, \cite{Kadom:Karp71}, \cite{Kudryash:Sine09}, 
\cite{Whitham}:%
\begin{equation}
\frac{\partial v}{\partial t}+v\frac{\partial v}{\partial x}=a\frac{\partial
^{2}v}{\partial x^{2}},\qquad a=\text{constant}>0  \label{vBurgers}
\end{equation}%
can be linearized by the \textit{Cole--Hopf substitution} \cite{Cole50}, 
\cite{Hopf50}:%
\begin{equation}
v=-\frac{2a}{u}\frac{\partial u}{\partial x},  \label{vColeHopf}
\end{equation}%
which turns it into the diffusion equation (\ref{sp1}). Solution of the
initial value problem has the form:%
\begin{equation}
v\left( x,t\right) =-\frac{a}{\sqrt{\pi at}}\frac{\partial }{\partial x}\ln %
\left[ \int_{-\infty }^{\infty }\exp \left( -\frac{\left( x-y\right) ^{2}}{%
4at}-\frac{1}{2a}\int_{0}^{y}v\left( z,0\right) dz\right) dy\right]
\label{vBurgersSolution}
\end{equation}%
for $t>0$ and suitable initial data on $%
%TCIMACRO{\U{211d} }%
%BeginExpansion
\mathbb{R}
%EndExpansion
.\medskip $

\bigskip

\textbf{Example 6\quad }Equation (\ref{vBurgers}) possesses a solution of
the form:%
\begin{equation}
v=F\left( x+Vt\right) ,\qquad V=\text{constant}
\end{equation}%
(we follow the original Bateman paper \cite{Bateman15} with slightly
different notations), if%
\begin{equation}
VF^{\prime }+FF^{\prime }=aF^{\prime \prime },
\end{equation}%
or%
\begin{equation}
\left( F+V\right) ^{2}\pm A^{2}=2aF^{\prime },
\end{equation}%
where $A$ is a positive constant. The solution is thus either%
\begin{equation}
v+V=A\tan \left[ \frac{A\left( x+Vt-c\right) }{2a}\right]
\end{equation}%
or%
\begin{equation}
\frac{A-v-V}{A+v+V}=\exp \left[ \frac{A}{a}\left( x+Vt-c\right) \right] ,
\end{equation}%
according as the $+$ or $-$ sign is taken. In the first case there is no
definite value of $v$ when $a$ tends to zero, while in the second case the
limiting value of $v$ is either $A-V$ or $A+V$ according as $x+Vt$ is less
or greater than $c.$ The limiting form of the solution is thus discontinuous 
\cite{Bateman15}.\medskip

Further examples can be found in Refs.~\cite{Craddock09}, \cite%
{Kudryash:Sine09}, \cite{Lop:Sus}, \cite{Miller77} and \cite{SuazoSusVega10}.

\section{Conclusion}

In this note, we have discussed connections of certain nonautonomous and
inhomogeneous diffusion-type equation and Burgers equation with solutions of
the Riccati and Ermakov-type systems that seem are missing in the available
literature. Traveling wave solutions of the Burgers-type equations are also
discussed.

\noindent \textbf{Acknowledgments.\/} We thank Professor Carlos Castillo-Ch%
\'{a}vez and Professor Carl Gardner for support, valuable discussions and
encouragement.

\end{document}